\documentclass[12pt]{iopart}

\pdfoutput=1

\expandafter\let\csname equation*\endcsname\relax
\expandafter\let\csname endequation*\endcsname\relax

\usepackage[english]{babel}
\usepackage{amsmath}
\usepackage{amssymb}
 \usepackage{amstext,amsthm,bbm}
\usepackage{graphicx}
\usepackage{color}

\usepackage{enumerate}
\usepackage{fullpage}

\newtheorem{theorem}{Theorem}[section]

\newtheorem{proposition}[theorem]{Proposition}

\theoremstyle{definition}

\newtheorem{remark}[theorem]{Remark}

\numberwithin{equation}{section}

\begin{document}

\title{\bf Topology trivialization transition in random non-gradient autonomous ODE's on a sphere}
\author{Y. V. Fyodorov }

\address{ King's College London, Department of Mathematics, London  WC2R 2LS, United Kingdom}
\begin{abstract}
We calculate the mean total number of equilibrium points in a system of $N$ random autonomous ODE's introduced by Cugliandolo et al. \cite{CKDP} to describe non-relaxational glassy dynamics on the high-dimensional sphere. In doing it we suggest a new approach which allows such a calculation to be done most straightforwardly, and is based on efficiently incorporating the Langrange multiplier into the Kac-Rice framework.
Analysing the asymptotic behaviour for large $N$ we confirm that the phenomenon of 'topology trivialization' revealed earlier for other systems holds also in the present framework with nonrelaxational dynamics. Namely, by increasing the variance of the random 'magnetic field' term in dynamical equations we find a 'phase transition' from the exponentially abundant number of equilibria down to just two equilibria. Classifying the equilibria in the nontrivial phase by stability remains an open problem.
\end{abstract}

\maketitle


\section{Introduction}
Time evolution of large complex systems is often described within the mathematical framework of coupled first-order autonomous nonlinear ordinary differential equations (ODEs)
\begin{equation}\label{1}
\frac{d {\bf x}}{dt} ={\bf F}({\bf x}), \quad {\bf x}\in \mathbb{R}_N
\end{equation}
The choice of the vector field ${\bf F}({\bf x})$ and hence detailed properties of the phase space trajectories strongly depend on specific applications and vary considerably from model to model. In reality, however, the detailed description of the vector fields is rarely available for large enough systems of considerable complexity in problems of practical interest.   At the same time it is natural to assume that in order to understand {\it generic} qualitative properties of the global dynamics of large systems of ODE's shared by many models of similar type it may be enough to retain only a few characteristic structural features of the vector field, treating the rest as random. In part such approach is methodologically inspired by undisputed success of the Random Matrix Theory (RMT) which manages to describe many properties of systems of very diverse nature, such as energy levels of heavy nuclei,  zeroes of the Riemann zeta-function and distances between tightly parked cars in a single conceptual framework \cite{RMT}. One also may note that large systems of random autonomous ODE's became popular in recent years in such fields as neuronal networks (see the classical paper \cite{SCS1988} and more recently \cite{WT,KS2015}), machine learning \cite{GF}, complex gene regulatory networks \cite{deJong}, populational dynamics of large ecosystems ( see e.g. \cite{TY2003,SVW05} for particular examples and  \cite{AlessinaRev} for a recent review of the long line of research
stemming from the classical paper \cite{May1972}), or random catalytic reaction networks \cite{SFM}.

Study of any dynamical system traditionally starts with the "local stability analysis" amounting to determining all possible points of equilibria and dynamical behaviour in the vicinity of those points. In practice this requires finding zeroes of the vector field  ${\bf F}({\bf x})$ and then classifying them by stability via eigenvalues of
the associated Jacobian matrix ${\partial F_i}/{\partial x_j}$, separately at each equilibrium. In general, the Jacobian matrices are asymmetric, unless the dynamics is of purely gradient descent type (also known as relaxational).
 Therefore generically equilbria are characterized by complex eigenvalues, with
locally stable equilibria (i.e. those attracting asymptotically nearby trajectories) being those characterized by eigenvalues with only negative real parts.
If however at least one of the eigenvalues has positive real part the equilibrium is locally unstable and the system will eventually go away from it along some directions in phase space. The detail of approaching equilibria or departing from it will certainly depend on the number of unstable directions and the size of imaginary parts of its eigenvalues. In fact, it is useful to have in mind that linearly unstable equilibria in systems with non-relaxational dynamics may give rise to limit cycles (stable or unstable) in nonlinear setting, and further to chaotic trajectories.

For  a dynamical system with many degrees of freedom performing the equilibria stability analysis for every and each equilbrium is a well-known formidable analytical and computational problem, and one may hope to get some insights into a generic behaviour by attempting to answer similar questions statistically for some classes of random models.
In such a framework it was recently discovered \cite{WT,FK2016} that the simplest nontrivial characteristic, the total number $\mathcal{N}_{tot}$ of all possible equilibria, may display as a function of some control parameters and the system size $N$ an abrupt transition from
a topologically trivial phase portrait with a single equilibrium into a topologically non-trivial regime characterised by an exponential in $N$ number of equilibria. Such phenomenon was then called a 'topology trivialization' \cite{FLD2013}, though the name  'topology detrivialization' can be considered as more appropriate.

In both cases the system of ODE's under consideration was of the form
\begin{equation}\label{2}
\dot{x}_i=-\mu x_i+f_i(x_1,\ldots, x_N)=0, \, i=1,\ldots, N.
\end{equation}
where the control parameter $\mu>0$ ensured the exponential relaxation towards zero equilibrium ${\bf x}=0$ in the absence of couplings $f_i({\bf x})$ between N degrees of freedom represented by variables $x_i$.  The coupling terms $f_i({\bf x})$ were however chosen in \cite{WT} and \cite{FK2016} in
somewhat different form dictated by particular context of the research in those papers. Namely, the paper \cite{WT} was motivated
by neuronal networks paradigm and hence the couplings were chosen in the form  $f_i({\bf x})=\sum_{j}J_{ij}S(x_j)$ where $S$ is an odd sigmoid function representing the synaptic nonlinearity and $J_{ij}$ are independent centred Gaussian variables representing the synaptic connectivity between neuron $i$ and $j$.  Although being Gaussian,  the corresponding vector field does not have a simple covariance structure as a function of ${\bf x}$ which makes a rigorous
mathematical analysis of the problem quite challenging. Nevertheless, a shrewd semi-heuristic analysis performed in \cite{WT} revealed the existence of critical coupling threshold  beyond which there is an exponential growth in the total number of equilibria.
In contrast, the work \cite{FK2016} was motivated by the context of famous "stability vs. diversity" debate due to Robert May \cite{May1972,AlessinaRev},
and hence used a freedom of choosing the nonlinear couplings $f({\bf x})$ as a random {\it homogeneous} Gaussian field which ensured
 mathematical tractability of the counting problem, and allowed to evaluate the expected value $\mathbb{E}\left\{\mathcal{N}_{tot}\right\}$ rigorously
 for any choice of parameters. The ensuing asymptotic analysis for $N\gg 1$ revealed then again the existence of an abrupt "topology detrivialization"
transition when decreasing the ratio of the relaxation rate $\mu$ to the variance of the Gaussian field beyond some critical threshold value $\mu_c$.
Interestingly, the rate of exponential growth of the total number of points close to the threshold turned out to be equal in both models,
 pointing towards certain universality of the observed transition.

 The problem of classifying the equilibria by stability in a general setting of \cite{FK2016} remains a major challlenge, apart from the special case of the {\it gradient descent} flow
 characterised by the existence of a potential function $V(\mathbf{x})$ such that $\mathbf{f}=-\nabla V$.  In that case the stable equilibria coincide with the stationary points of the Lyapunov function $L(\mathbf{x})=\mu{|{\bf x}|^2}/{2}+V({\bf x})$, and
 finding mean number of the minima (and indeed saddle points with any given number of negative eigenvalues of the Hessian) is possible, see \cite{FyoNad} for the original calculation and \cite{my2015} and references therein for a recent review. The mean number of stable equilibria for a general non-potential flow was very recently addressed in \cite{BFKunpub} using large deviation ideas.  It turned out that generically when  decreasing $\mu$ below the threshold value $\mu_c$ the system first transits over from the '{\it absolute stability'} regime with a single stable equilibrium for $\mu>\mu_c$ to '{\it absolute instability}' regime for $\mu_B<\mu<\mu_C$ where equilibria are exponentially abundant, but typically {\it all} of them are unstable. Finally, at even smaller relaxation rate $\mu<\mu_B$ stable equilibria become exponentially abundant, but their fraction to totality of all equilibria remains exponentially small.

 Despite apparent advantages ensuring analytical tractability, a certain drawback of the model choice in \cite{FK2016} is that
 to generate many samples of homogeneous (i.e. stationary and isotropic) random fields even in moderately high dimensions is known to be prohibitively expensive numerical procedure. This fact makes numerical investigation of the corresponding dynamical systems problematic, and hence the
 computational verification of the theoretical predictions hardly feasible. This motivated us to consider a model of systems of random ODE's whose behaviour shows the same rich phenomenology as those in \cite{WT,FK2016}, but whose vector fields can be relatively easily generated while retaining the property of the equilibria counting problem being analytically tractable.  We will see that the so-called "spherical model with non-relaxational dynamics" introduced in \cite{CKDP} (see also \cite{CriSomp} for earlier related studies) satisfies all these requirements, and is therefore a convenient framework for our goals.

 In this paper we are going to consider the following system of coupled ODE's:
 \begin{equation}\label{3}
 \dot{x}_k=-\lambda(t)x_k+h_{k}+f_k({\bf x}), \quad k=1,\ldots,N
 \end{equation}
where the (time dependent) Langrange multiplier $\lambda(t)$ is introduced to ensure that at any moment of time the $N-$component state vector ${\bf x}=(x_1,\ldots,x_N)^T$ satisfies the "spherical constraint"
 \begin{equation}\label{4}
 \sum_{k=1}^N x_k^2=N
 \end{equation}
 so that the dynamics is confined to the surface of $N-1$ dimensional sphere of the radius $R=\sqrt{N}$ centered at zero.
In particular, by differentiating the constraint (\ref{4}) and using (\ref{3}) one straightforwardly finds the Lagrange multiplier explicitly in the form
 \begin{equation}\label{5}
 \lambda(t)=\frac{1}{N}\sum_{k=1}^N x_k \left(h_k+f_k({\bf x})\right)
 \end{equation}
 If in (\ref{3}) one sets all the couplings $f_k({\bf x})$ to zero, the vector ${\bf h}=(h_1,\ldots,h_N)$ (following the terminology in the spin glass area we will frequently call ${\bf h}$ the 'magnetic field') drives the exponential relaxation of the system towards the global stable equilibrium at ${\bf x}=\sqrt{N} {\bf h}/|{\bf h}|$. From that point of view the role of the magnetic field ${\bf h}$ is analogous to the role of the parameter $\mu$ in the dynamics described by the system (\ref{2}).
 Note however that the topology of the sphere excludes the possibility of a single stable equilibrium for a vector field, and indeed
 total number of equilibria in the present model is two rather than one, with second totally unstable equilibrium at ${\bf x}=-\sqrt{N} {\bf h}/|{\bf h}|$.
 Our goal is to investigate how by introducing a strong enough random coupling fields $f_k({\bf x})$ this "trivial" topology of the phase portrait will be replaced with a very rich phase portrait with abundance of equilibria.

\section{Model definition and main results}

To specify the model further in a random setting, we have to make some statistical assumptions about components of the vector field $f_k({\bf x})$.
To that end following \cite{CKDP} we consider those components to be Gaussian mean zero random fields, with the covariance structure given by
\begin{equation}\label{6}
\mathbb{E}\left\{f_k({\bf x})f_p({\bf x'})\right\}=\delta_{kp}\Phi_1\left(\frac{\bf{x}\cdot{\bf x}'}{N}\right)+\frac{x_p  x_k'}{N}\Phi_2\left(\frac{\bf{x}\cdot {\bf x}'}{N}\right)
 \end{equation}
where here and henceforth we will use ${\bf a}\cdot {\bf b}$ to denote the inner product of vectors ${\bf a}$ and ${\bf b}$ (and further denote ${\bf x}^2={\bf x}\cdot{\bf x}$), and $\delta_{kp}$ stands for the Kronecker delta and we assume that $\Phi_1(u), \Phi_2(u)$ are some $N-$independent functions of the real variable $u$
satisfying
\begin{equation}\label{6a}
0 < \Phi_1\left(1\right)\le \Phi_1'\left(1\right) , \quad\quad -\Phi_1(1)\le \Phi_2(1)\le \Phi_1'(1)
\end{equation}
We further assume that the components $h_i, \, i=1,\ldots,N$ of the 'magnetic field' vector ${\bf h}$ are random real i.i.d. mean zero Gaussian variables:
 \begin{equation}\label{6b}
\mathbb{E} \left\{h_ih_j\right\}=\delta_{ij}\sigma^2, \quad \forall i,j
 \end{equation}

To motivate the choice (\ref{6})-(\ref{6a}) it is instructive to consider a particular representative example of fields of such type which can be constructed explicitly as follows. Define
\begin{equation}\label{7}
f_k({\bf x})=\sum_{j=1}^N J^{(1)}_{kj}x_j+\sum_{n,m=1}^N J^{(2)}_{knm}x_nx_m, \quad k=1,\ldots, N
 \end{equation}
where the coefficients $J^{(1)}_{kj}$ and $J^{(2)}_{knm}$ are further represented as
\begin{equation}\label{7a}
J^{(1)}_{kj}=V^{(1)}_{kj}+\alpha_1 V^{(1)}_{jk},  \quad J^{(2)}_{knm}=V^{(2)}_{knm}+\alpha_2 \left(V^{(2)}_{nkm}+V^{(2)}_{nmk}\right)
 \end{equation}
with real parameters $\alpha_1,\alpha_2$. Finally, we choose $N^2$ real variables $V^{(1)}_{kj}$ to be random mean zero i.i.d. Gaussians
 so that the corresponding covariance structure is given by:
\begin{equation}\label{7b}
\mathbb{E}\left\{ V^{(1)}_{kj} V^{(1)}_{pm}\right\}=\frac{J_1^2}{N} \delta_{kp}\delta_{jm}, \quad J_1>0
 \end{equation}
and similarly choose $N^3$ variables $V^{(2)}_{knm}$ to be mean zero real i.i.d. Gaussian variables independent of  $V^{(1)}_{kj}$
 satisfying
\begin{equation}\label{7c}
\mathbb{E} \left\{V^{(2)}_{knm} V^{(2)}_{pqr}\right\}=\frac{J_2^2}{N^2} \delta_{kp}\delta_{nq} \delta_{mr}, \quad J_2\ge 0
 \end{equation}

It is now straightforward to calculate the covariances of the fields from (\ref{7}) and to find it is given precisely by the form (\ref{6}) with
\begin{equation}\label{8}
\Phi_1\left(u\right)=(1+\alpha_1^2)J_1^2 u+(1+2\alpha_2^2)J_2^2 u^2, \quad
\Phi_2\left(u\right)=2\alpha_1J_1^2+2\alpha_2(2+\alpha_2)J_2^2 u
 \end{equation}
 Note that
 \begin{equation}\label{8a}
\Phi_1'\left(1\right)-\Phi_1\left(1\right)=(1+2\alpha_2^2)J_2^2\ge 0, \quad \Phi_1'\left(1\right)-\Phi_2\left(1\right)=(1-\alpha_1)^2J_1^2 +2(1-\alpha_2)^2J_2^2\ge 0
 \end{equation}
 as well as
\[
\Phi_1\left(1\right)+\Phi_2\left(1\right)=(1+\alpha_1)^2J_1^2+ (1+2\alpha_2)^2J_2^2\ge 0
\]
 so indeed (\ref{6a}) holds. This construction can be easily extended to include in (\ref{7}) polynomials terms of any higher order \cite{CKDP}, which motivates one to work with the general form of the covariance specified in (\ref{6}) via two functions $\Phi_1(u)$ and $\Phi_2(u)$
 satisfying (\ref{6a}).

\begin{remark} Suppose that the vector field with components $f_k({\bf x})$ describes a {\it gradient descent} dynamics, that is $f_k({\bf x})=\frac{\partial V({\bf x})}{\partial x_k}$.
Assume further that the "potential" function   $V({\bf x})$ is a gaussian isotropic random field with the covariance structure specified as
\begin{equation}\label{9}
\mathbb{E} \left\{V({\bf x}) V({\bf x}')\right\}=N\,F_V\left(\frac{\bf{x}\cdot{\bf x}'}{N}\right)
 \end{equation}
Then it is easy to check that the covariance structure of the  fields $f_k({\bf x})$ is exactly of the form (\ref{6}), with
functions $\Phi_1\left(u\right)$ and $\Phi_2\left(u\right)$ being related as
\begin{equation}\label{9a}
\Phi_1\left(u\right)=F'_V\left(u\right), \quad \Phi_2\left(u\right)=\Phi_1'\left(u\right)
 \end{equation}
 \end{remark}

 \vspace{0.5ex}

As follows from (\ref{8a}) for the particular choice of the field (\ref{7})-(\ref{7c}) the second relation in (\ref{9a}) can be satisfied only for $\alpha_1=\alpha_2=1$.
And indeed, one can check that such a special choice ensures that the field  is gradient, with the potential $V({\bf x})$ given by
\begin{equation}\label{9c}
V({\bf x})=\sum_{k,j=1}^N\quad  V^{(1)}_{kj}x_kx_j+\sum_{k,n,m=1}^N V^{(2)}_{knm}x_kx_nx_m,
 \end{equation}
so that $F_V\left(u\right)=J_1^2 u^2+J_2^2 u^3$.

In fact, the potential defined in (\ref{9c}) is a particular representative of the energy functionals associated with the so-called
spherical model of spin glasses whose dynamical and equilibrium static properties keep attracting over the last decades a lot of attention both in physical \cite{CS1992,CHS1993} and mathematical \cite{BDG1,Talagrand2003,CHHS} literature. It is evident that for this special case, after setting the magnetic field ${\bf h}$ to zero, the equilibria of the gradient dynamics on the sphere coincide with the stationary points of the potential $V({\bf x})$.
The associated system of ODE's (\ref{3}) describes   'relaxational' gradient descent towards the global minimum of the energy functional.  The problem of finding the mean number of stationary points with a given index for isotropic Gaussian potentials with covariance (\ref{9}) constrained to the sphere  was originally solved in the insightful papers \cite{Auf1,Auf2}, and later revisited for ${\bf h}\ne 0$ in \cite{FLD2013,my2015} from the point of view of concentrating on the 'topology trivialization' phenomenon in that framework.
 Those works demonstrated that calculating the mean number of stationary points in such a setting can be mapped onto a random matrix problem related to the standard GOE ensemble, and in this respect remains quite similar to the case of stationary fields in the Euclidean space, where such reduction was discovered originally, see \cite{Fyo04,BD07,FyoWi07,FyoNad}. Note also a recent work \cite{CheSchw} giving a unified treatment of both cases which went beyond certain restrictions in the original papers.

For a general choice of the functions $\Phi_1(u)$ and $\Phi_2(u)$ in (\ref{6}) (e.g. choosing values $\alpha_1,\alpha_2$ in (\ref{7a}) different from unity)  the corresponding vector field is not gradient, and the associated dynamics is not relaxational, which was precisely the idea behind introducing a variant of that model in \cite{CKDP}.  Our goal is to solve the problem of calculating the mean number of equilibria for such type of  general non-relaxational dynamics on the sphere. In doing it we will develop a new method which allows such a calculation to be done most straightforwardly, and is based on efficiently incorporating the Langrange multiplier into the Kac-Rice framework. Below we give a summary of the main results of the paper.

Our main result is the following
\begin{theorem}\label{Thm1}
 Define the Gaussian Elliptic Ensemble of $N\times N$ random real Gaussian matrices $X_N$ whose entries $X^{(N)}_{ij}$ have zero mean and covariance
$\left\langle X^{(N)}_{ij}X^{(N)}_{nm} \right\rangle=\delta_{in}\delta_{jm}+  \tau \delta_{jn}\delta_{im}$ with a real parameter $\tau\in [-1,1]$.
 Equivalently, the corresponding joint probability density of $X_{N}$ is given by
\begin{equation}\label{11}
\mathcal{P}(X_N=X)=\mathcal{Z}_N^{-1}\exp\Big[-\frac{1}{2(1-\tau^2)}\Tr \big( X X^T-\tau X^2\big) \Big],
\end{equation}
where $\mathcal{Z}_N$ is the associated normalization constant
\[
\mathcal{Z}_N=2^{N/2}\pi^{N(N+1)/2}(1+\tau)^{N(N+1)/4}(1-\tau)^{N(N-1)/4} \, .
\]
Denote by $\rho_{N}^{(r)}(x)$ the density of real eigenvalues of such matrices averaged over all realisations of $X_N$. It is convenient to normalize $\rho_{N}^{(r)}(x)$  in such a way that  $\int_{\alpha}^{\beta} \rho^{(r)}_{N}(x)\,dx$ gives the average number of real eigenvalues of $X$ in the interval $[\alpha,\beta]$. Then the mean number of the points of equilibria for the associated dynamical system (\ref{3}) with the Gaussian vector field described by
(\ref{6})-(\ref{6a}) and ${\bf h}$ described by (\ref{6b}) is given by
\begin{equation}\label{12}
\mathbb{E}\left\{\mathcal{N}_{tot}\right\}=2\sqrt{N\frac{1+\tau}{b^2+\tau}}\, b^{1-N} \int_{-\infty}^{\infty}e^{-\frac{N}{4}\lambda^2\,B}\rho^{(r)}_{N}(\lambda \sqrt{N})\,d\lambda
\end{equation}
where
\begin{equation}\label{12a}
\tau=\frac{\Phi_2(1)}{\Phi_1'(1)}, \quad b^2=\frac{\sigma^2+\Phi_1(1)}{\Phi_1'(1)}, \quad B=\frac{2}{1+\tau}\,\frac{1-b^2}{b^2+\tau}
\end{equation}

\end{theorem}

\begin{remark}
Inequalities (\ref{6a}) ensure that the parameter $\tau$ defined in (\ref{12a}) satisfies $-1\le \tau\le 1$, and further $b^2+\tau\ge 0$. It is clear from (\ref{9a})
that purely gradient 'relaxational' dynamics always corresponds to the value $\tau=1$, so that the value of $\tau$ can be used to control
the degree to which the dynamics described by the system of ODE's (\ref{2}) is not-relaxational. As can be inferred
from the example (\ref{7})-(\ref{7a}) the value $\tau=-1$ is exceptional and only possible for linear ($J_2=0$) purely antisymmetric $\alpha_1=-1$ fields.
In what follows we will assume $\tau>-1$.  Similarly, the case $b^2+\tau=0$ is exceptional and will be excluded as well, making the parameter $B$ in (\ref{12a}) well-defined.
\end{remark}
\begin{remark}
We shall see that replacing in the right hand side of (\ref{12})  the integration domain $\lambda\in (-\infty,\infty)$ by an interval $\lambda\in [\alpha/\Phi_1'(1),\beta/\Phi'(1)]$ gives the mean number of the points of equilibria such that the associated Lagrange multipliers (\ref{5}) take values in the interval $[\alpha,\beta]$.
\end{remark}

\vspace{0.5cm}

The representation (\ref{12}) is valid for any $N\ge 1$ but naturally we are mostly interested in the asymptotics for large systems, that is $N\gg 1$.
 The asymptotic behaviour of the  density $\rho^{(r)}_{N}(x)$ is well-understood \cite{ForNag2008}, see also \cite{FK2016}, which allows us to find the leading asymptotic behaviour of $\mathbb{E}\left\{\mathcal{N}_{tot}\right\}$ for given $N-$independent values of the parameters $\tau$ and $b$.
 To make our analysis most simple we will assume in the rest of the paper that $N$ is an even integer.
 We will show that
 \begin{proposition}\label{Prop1}
 For any fixed $|\tau|<1$ and $b<1$ asymptotically as $N\gg 1$ holds
  \begin{equation}\label{13}
\mathbb{E}\left\{\mathcal{N}_{tot}\right\}=2\sqrt{\frac{1+\tau}{1-\tau}}\, \frac{b}{\sqrt{1-b^2}}\, \exp{\left(N\ln{\frac{1}{b}}\right)}\,,
\end{equation}
whereas for $b>1$ and any $-1<\tau\le 1$
\begin{equation}\label{13a}
\lim_{N\to \infty} \mathbb{E}\left\{\mathcal{N}_{tot}\right\}=2
\end{equation}
\end{proposition}

The above behaviour describes precisely the phenomenon of the "topology trivialization"  \cite{FLD2013,my2015}, with the role of the control parameter played by the growing variance  of the magnetic field $\sigma$. Namely, using the definition of the parameter $b$ in (\ref{12a}) and
the first of the conditions (\ref{6a}) we see that for $\sigma<\sigma_c=\sqrt{\Phi_1'\left(1\right)-\Phi_1\left(1\right)}$ the equilibria are generically exponentially abundant, whereas for $\sigma>\sigma_c$ their number abruptly drops down to the minimal possible value $2$. From this point of view one may
speak about a  'topology trivialization phase transition' taking palce at the critical value $\sigma=\sigma_c$.
One can further analyse how this transition happens on a finer scale for large but finite values of $N$ by "zooming in" the appropriately scaled critical regime around  $\sigma=\sigma_c$.
Such an analysis reveals that the transition happens in two stages: first in the regime $|\sigma_c-\sigma|\sim N^{-1}$
the mean number of equilibria drops from exponentially many down to values of order of $\sqrt{N}$, and then for even larger fields such that $\sigma>\sigma_c$
and $\sigma-\sigma_c\sim N^{-1/2}$ the number further drops from  values of order of $\sqrt{N}$ to values of order of unity. This behaviour is summarized in
the following two propositions:
 \begin{proposition}\label{Prop2}
 Fix any $|\tau|<1$ and for $N\to \infty$ scale the parameter $b$ with $N$ as $b^2=1-\frac{\gamma}{N}$, with the parameter $-\infty<\gamma<\infty$ being fixed.
 Then
  \begin{equation}\label{14}
\lim_{N\to \infty} \frac{1}{\sqrt{N}}\mathbb{E}\left\{\mathcal{N}_{tot}\right\}=4\sqrt{\frac{1}{2\pi}}\sqrt{\frac{1+\tau}{1-\tau}}\, e^{\frac{\gamma}{2}}\int_0^1e^{-\frac{\gamma}{2}\lambda^2}d\lambda
\end{equation}
\end{proposition}
and
\begin{proposition}\label{Prop3}
Fix any $|\tau|<1$ and for $N\to \infty$ scale the parameter $b$ with $N$ as $b^2=1+\frac{\kappa}{\sqrt{N}}$, with the parameter $\kappa>0$ being fixed. Then
\begin{equation}\label{15}
\lim_{N\to \infty} \mathbb{E}\left\{\mathcal{N}_{tot}\right\}=4e^{-\frac{\tilde{\kappa}^2}{4}}\int_{-\infty}^{\infty} e^{\tilde{\kappa}\zeta}
\rho^{(r)}_{edge}(\zeta)\,d\zeta
\end{equation}
where we defined $\tilde{\kappa}=\kappa\sqrt{\frac{1-\tau}{1+\tau}}$ and
\begin{equation}\label{15a}
\rho_{edge}^{(r)}(\zeta)=\frac{1}{2\sqrt{2\pi}}\left\{
erfc (\sqrt{2}\zeta) +\frac{1}{\sqrt{2}}e^{-\zeta^2}\left[1+erf\left(\zeta\right)\right]\right\}
\end{equation}
with $erf(x)=1-erfc(x)= \frac{2}{\sqrt{\pi}}\int_0^xe^{-t^2}\,dt$.
\end{proposition}
These two formulas fully describe the crossover between (\ref{13}) and (\ref{13a}). Indeed, consider the parameter $\kappa$ in (\ref{15}) to be large: $\kappa\gg 1$. It is easy to see that in such a limit the integral in (\ref{15}) is dominated by the region of large positive $\zeta \gg 1$ where we can use the asymptotics $\rho_{edge}^{(r)}(\zeta)\approx \frac{1}{2\sqrt{\pi}} e^{-\zeta^2}$. The integral then becomes effectively Gaussian and one immediately finds that the right-hand side of (\ref{15}) tends to the value $2$ thus matching the value in (\ref{13a}). On the other hand, consider  the parameter $\gamma$ in (\ref{14}) to be positive and large: $\gamma\gg 1$. The integral over $\lambda$ is then approximately equal to $\sqrt{\pi/2\gamma}$ and the resulting expression exactly matches (\ref{13}) after replacing $b^2=1-\frac{\gamma}{N}$ in the latter.

Finally, let us satisfy ourselves that the limit $\gamma \to -\infty$ in (\ref{14}) exactly matches the limit $\kappa \to 0$ in (\ref{15}) so that indeed
all possible regimes are covered by these two equations. It is easily seen that for the big negative $\gamma$ the equation  (\ref{14}) implies for the mean number of critical points
\begin{equation}\label{16}
\mathbb{E}\left\{\mathcal{N}_{tot}\right\}\approx 4\sqrt{\frac{1}{2\pi}}\sqrt{\frac{1+\tau}{1-\tau}}\, \frac{\sqrt{N}}{|\gamma|}
\end{equation}
On the other hand, for $\kappa \to 0$ in (\ref{15}) the integral is dominated by the values $\zeta \to -\infty$ where $\rho_{edge}^{(r)}(\zeta)\approx \frac{1}{\sqrt{2\pi}}$. Evaluating that integral asymptotically one indeed reproduces (\ref{16}) after identifying $|\gamma|=\kappa\sqrt{N}$ as expected.

Finally, one may notice that the asymptotics (\ref{13}) is only valid for $\tau<1$ and needs to be modified when $\tau\to 1$.
Recall that the value $\tau=1$ corresponds to purely gradient flows. The gradient limit $\tau=1$ can be approached by scaling $\tau$ with $N$ appropriately and in such a regime of a  weakly non-gradient flow
the asymptotics (\ref{13}) is replaced by the following:
\begin{proposition}\label{Prop4}
 Let $\tau=1-\frac{u^2}{N}, \, 0\le u<\infty$ and consider $b^2<1$. Then for large $N\gg 1$ asymptotically
\begin{equation}\label{17}
\mathbb{E}\left\{\mathcal{N}_{tot}\right\} = 4 \, e^{N\ln{\frac{1}{b}}} \,\sqrt{\frac{2N b^2}{\pi(1-b^2)}}\int_0^{1}\,e^{-u^2p^2} dp \,
 \end{equation}
 \end{proposition}
 Note that in this limit the parameter $B$ defined in (\ref{12a}) is given by $B=\frac{1-b^2}{1+b^2}$ so that the above equation can be rewritten as
 \begin{equation}\label{17a}
\mathbb{E}\left\{\mathcal{N}_{tot}\right\} = 4 \, e^{\frac{N}{2}\ln{\frac{1+B}{1-B}}} \,\sqrt{\frac{N (1-B) }{\pi B}}\int_0^{1}\,e^{-u^2p^2} dp \,
 \end{equation}
  which in the pure gradient limit $u=0$ matches exactly the expression (60 ) from the paper \cite{my2015}.

\subsection{Discussion and open problems}

The results for counting the totality of equilibria for a random dynamical system (\ref{3}) describing generically (i.e. for $\tau\ne 1$) non-relaxational
dynamics on a high-dimensional sphere shows precisely the same type of behaviour as the system (\ref{2}) analyzed earlier in \cite{FK2016} in the Euclidean setting. In particular, the two models display a very similar 'topology trivialization' transition with changing the appropriate control parameter, the single-site relaxation rate in the Euclidean setting and the magnitude of the magnetic field in the spherical case.
From that point of view the two models belong essentially to the same universality class. The advantage of working on the sphere is that such systems can
be much easier simulated numerically for moderate to large values of the parameter $N$. For the gradient case of spherical model a numerical study of topology trivialization and related aspects were undertaken in \cite{Mehta1,Mehta2}. Note also general interest in landscape explorations for optimization and learning problems \cite{SGBL}; in that context the phenomenon of topology trivialization
was recently found to be of relevance for training 'deep learning' networks \cite{ChaudSoatto,Julius}.

The problem of classifying  each and every equilibrium point into locally stable or unstable seems a hard task.  Given the stochastic setup of our model the question about stability of individual equilibria  may be even the wrong question to ask, whereas addressing the statistics of the {\it number} of stable equilibria seems very appropriate. Unfortunately, the framework using the Lagrange multipliers, which as we have shown works very efficiently when counting all the stationary points of the dynamics, is not immediately adjustable to the problem of counting equilibria with prescribed number of unstable directions, or even only the stable equilibria. Note that for the simplest case $h_{k}=0$ and linear potential forces $f_k({\bf x})=-\sum_jJ_{kj}x_j, \, J_{jk}=J_{kj}$, the dynamics (\ref{3}) drives the time-dependent Lagrange multiplier (\ref{5}) towards the ever-smaller values, and the stable equilibrium corresponds to the value of $\lambda$ given by the smallest eigenvalue of the symmetric matrix $J_{kj}$.  Given that the methods of the present paper allow for counting the stationary points with values of the Lagrange multipliers in any given interval, an interesting question is to what extent (and if at all, under what conditions) stability of equilibria of a general nonlinear and non-potential system on the sphere may be judged by the values of its Lagrange multipliers.

 On the other hand, it should be possible to generalize approaches of \cite{my2015,Auf1,Auf2} to the non-potential dynamics and arrive
 at the ensemble average of the total number of stable equilibria, $\mathcal{N}_{stab}$, over all realisations of the vector field in terms of a random matrix integral.  In the limiting case of a purely gradient dynamics $\tau=1$ that integral was related to the probability density of the maximal eigenvalue of the GOE matrix \cite{Auf1,FyoNad}, with the latter being a well-studied object in the random matrix theory. This observation was used to evaluate $\mathbb{E}\left\{\mathcal{N}_{stab}\right\}$ for $N\gg 1$, see \cite{my2015} and references therein. One finds  that  $\mathbb{E}\left\{\mathcal{N}_{stab}\right\} \to 1$ if $b>1$, whilst if $b<1$ then, to the leading order in $N$,  the stable equilibria are still exponentially abundant, but their number is a vanishing fraction of the total number of equilibria. Thus, in the case of purely gradient dynamics large nonlinear autonomous systems assembled at random undergo an abrupt change from a typical phase portrait with a single stable equilibrium to a phase portrait dominated by an exponential number of unstable equilibria with an admixture of a smaller, but still exponential in $N$, number of stable equilibria.
 As was already mentioned in the introduction, in the case of a generic random non-potential flow in the Euclidean space  the mean number of stable equilibria was very recently analyzed in \cite{BFKunpub} using large deviation ideas.  It turned out that generically when  decreasing $\mu$ below the threshold value $\mu_c$ the system first transits from the '{\it absolute stability'} regime with a single stable equilibrium for $\mu>\mu_c$ over to '{\it absolute instability}' regime for $\mu_B<\mu<\mu_c$ where equilibria are exponentially abundant, but typically {\it all} of them are unstable. Finally, at even smaller relaxation rate $\mu<\mu_B$ stable equilibria become exponentially abundant, but their fraction to totality of all equilibria remains exponentially small. It is therefore an interesting question to adjust those methods to the present
 model on the sphere in the general case of non-gradient dynamics $\tau<1$.

  Although our investigation is concerned with the ensemble average of the number of equilibria  we expect that in the limit $N\to \infty$ the deviations of $\mathcal{N}_{tot}$ from its average $\mathbb{E}\left\{\mathcal{N}_{tot}\right\}$ are relatively small. The problem of estimating the deviation of   $\mathcal{N}_{tot}$  from its average value in the exponential abundance regime is an open and interesting question. In this context we  would like to mention the recent work of Subag \cite{Subag2015} who proved that  in the gradient  spherical model the deviations of  $\mathcal{N}_{tot}$ from its average are negligible in the limit of large system size.

   One may hope that numerical simulations of the model should be feasible, and understanding its non-relaxational dynamics rigorously in general setting
   is a challenging issue deserving further investigations, see \cite{CKDP,CriSomp}. Let us finally mention that the simplest, yet not fully trivial case is obtained by retaining only the first, linear, term in (\ref{7}). The case (referred to in the spin glass studies as $p=2$ case) is indeed very special, though its static and dynamic properties enjoyed over the years thorough attention, starting from the classical work \cite{KTJ} through the later papers
   \cite{CD1,CD2,ZKH} to the most recent results in \cite{FLD2013,DZ,BL2015,GT,FPS}. In particular, in this case  $ \Phi'(1)=\Phi(1)$ so that for $\sigma=0$  the parameter $b=1$. Hence for any fixed $\sigma>0$ we will have  $b>1$ implying only two points of equilibria and, as a consequence, a trivial exponential relaxation \cite{CD2}. To make the dynamics less trivial one needs to scale $\sigma\sim N^{-1}$ so that for $\tau<1$ the number of equilibria will be of  the order of $\sqrt{N}$ according to {\bf Proposition (\ref{Prop2})}.  An interesting question is then if this number of equilibria is already enough to support a nontrivial dynamics with aging affects \cite{CD1,ZKH,FPS} typical for the gradient counterpart of the problem without magnetic field, where the number of equilibria is of the order of $N$. The same scaling still holds for small magnetic field appropriately scaled with $N$ \cite{FLD2013}
     and is expected to hold for weakly non-relaxational dynamics like that considered in   {\bf Proposition (\ref{Prop4})} but with the parameter $b\to 1$ when $N\to \infty$. This limit is not covered by  {\bf Proposition (\ref{Prop4})} and its study is left for future research together with building extension of the analysis of \cite{FPS} to the non-relaxational case.

{\bf Acknowledgements}. The author is grateful to L. Cugliandolo for drawing his attention to \cite{CKDP} where the model studied in this paper was introduced and for further discussions on constrained dynamics.  The financial support by  EPSRC grant  EP/N009436/1 "The many faces of random characteristic polynomials" is acknowledged with thanks.

\section{Kac-Rice formulae within Lagrange multiplier framework: proof of the Theorem (\ref{Thm1})}

According to its very definition, every equilibrium point of the system (\ref{3}-\ref{4}) is associated with an $N+1$ component vector $\left(\begin{array}{c} {\bf x}\\ \lambda\end{array}\right)$,  with the values of the state vector ${\bf x}\in \mathbb{R}_N$ and the (time-independent) Lagrange multiplier
$-\infty<\lambda<\infty$  being chosen to  solve the system of $N+1$ algebraic equations
 \begin{equation}\label{18}
 -\lambda\, x_k+h_{k}+f_k({\bf x})=0, \quad k=1,\ldots,N \quad \mbox{and} \quad  \sum_{k=1}^N x_k^2=N
 \end{equation}

The non-linear system (\ref{18}) may have multiple solutions
whose number and locations depend on the realisation of the
random field {\bf f(x)} and parameters $h_k$. It is well-known that different solutions of such a system are generically, with probability one, isolated points in $\mathbb{R}^{N+1}$. Counting only solutions such that the corresponding values of the Lagrange multipliers $\lambda$ belong to an interval $[\alpha,\beta]$,  the total number of the solutions is given by the Kac-Rice metateorem, see e.g. \cite{math2},
 \begin{equation}\label{19}
\mathcal{N}_{tot} \!=\!\int_{[\alpha,\beta]} d\lambda \int_{\mathbb{R}^N} \!\! \left| \det\left( \begin{array}{cc} \frac{\partial f_k}{\partial x_l} -\lambda \delta_{kl} & -{\bf x} \\ 2{\bf x}^T & 0 \end{array}
\right) \right|\,\,\delta\left( \sum_{k=1}^N x_k^2-N\right) \prod_{k=1}^N
\delta\left(-\lambda \, x_k +h_k +\! f_k(\mathbf{x})\right) \,dx_k\, ,
 \end{equation}
where $\delta(u)$ stands for the Dirac $\delta$-distribution, and the $(N+1)\times (N+1)$ matrix under the determinant sign in the integrand is the Jacobian associated with the system (\ref{18}).

The number $\mathcal{N}_{tot}$ changes from one realization of the random field to another,
and the ultimate goal of the theory should be providing the distribution
of $\mathcal{N}_{tot}$. In such generality the problem is however very challenging, and in the present paper we restrict ourselves to the simplest nontrivial characteristics, the mean value $\mathbb{E}\left\{\mathcal{N}_{tot}\right\}$ which is given in {\bf Theorem (\ref{Thm1})}.
Proving it requires in the first place finding a way of performing the averaging of (\ref{18}) over the joint probability density $\mathcal{P}_{\bf{x}}\left({\bf f},K\right)$ of random vector ${\bf f}$ with $N$ components $f_k(\mathbf{x})$ and random $N\times N$ matrix $K$ with entries $K_{kl}=\frac{\partial f_k}{\partial x_l}$ where the derivatives are taken at the same point ${\bf x}$, as well as over the 'magnetic fields' $h_k$. Such an averaging can be indeed performed and its result is summarized in the following

\begin{theorem}\label{Thm2}
 Define the Gaussian Elliptic Ensemble of $N\times N$ random real Gaussian matrices $X_{N-1}$  as in (\ref{11}), but with the shift $N\to N-1$, and
  further define the parameters $\tau$ and $b^2$ as in (\ref{12a}). Assume that the Gaussian vector field is described by
(\ref{6})-(\ref{6a}) and ${\bf h}$ described by (\ref{6b}).
  Then the mean number of the points of equilibria  for the associated dynamical system (\ref{3}) such that the associated Lagrange multipliers (\ref{5}) take values in an interval $[\alpha,\beta]$ is given by
\begin{equation}\label{A}
\mathbb{E}\left\{\mathcal{N}_{tot}\right\}= \frac{1}{2^{\frac{N}{2}-1}\Gamma\left(\frac{N}{2}\right)}\frac{1}{\sqrt{b^2+\tau}} \, b^{1-N}
\int_{\alpha/\sqrt{\Phi_1'(1)}}^{\beta/\sqrt{\Phi_1'(1)}} e^{-N\frac{\lambda^2}{2(b^2+\tau)}\,}
\mathbb{E}\left\{\left| \det\left( \begin{array}{c} X -\lambda\sqrt{N}{\bf 1}_{N-1}   \end{array}
\right) \right|\right\}_{X_{N-1}}\, d\lambda
\end{equation}

\end{theorem}

\begin{proof}

Averaging over the 'magnetic fields' $h_k$ is straightforward to perform as (\ref{6b}) implies
\begin{equation}\label{20}
\mathbb{E}\left\{ \prod_{k=1}^N
\delta\left(-\lambda \, x_k +h_k +\! f_k(\mathbf{x})\right)\right\}=\frac{1}{(2\pi\sigma^2)^{N/2}}\exp{\left\{-\frac{1}{2\sigma^2}\left(\lambda{\bf x}-\bf{f}({\bf x})\right)^2\right\}}
\end{equation}
so that
\begin{equation}\label{20a}
\mathbb{E}\left\{\mathcal{N}_{tot}\right\} \!=\!\int_{[\alpha,\beta]} d\lambda \int_{\mathbb{R}^N} \, d{\bf x}\, \delta\left( {\bf x}^2-N\right) \,\mathcal{I}(\bf{x},\lambda)
\end{equation}
where we have defined
\begin{equation}\label{20b}
\mathcal{I}({\bf x},\lambda) = \int \int \,\! \mathcal{P}_{\bf{x}}\left({\bf f},K\right) e^{-\frac{1}{2\sigma^2}\left(\lambda{\bf x}-{\bf f}\right)^2} \left| \det\left( \begin{array}{cc} K -\lambda{\bf 1}_N   & -{\bf x} \\ 2{\bf x}^T & 0 \end{array}
\right) \right|\, dK\, d{\bf f} \, ,
\end{equation}
where ${\bf 1}_N$ stands for $N\times N$ identity matrix,

To deal with the ensemble average we find it most convenient to introduce the Fourier-transform of the joint probability density via
\begin{equation}\label{21}
\mathcal{P}_{\bf{x}}\left({\bf f},K\right)=\int \mathcal{F}_{\bf{x}}({\bf q},Q) e^{-i{\bf q}^T{\bf f}-i\mbox{\small Tr}\left(KQ\right)} \frac{d{\bf q}}{(2\pi)^N}\frac{d Q }{(2\pi)^{N^2}}
\end{equation}
where ${\bf q}$ and $Q$ are $N-$component real vector and $N\times N$ real matrix, respectively, and $\mathcal{F}_{\bf{x}}({\bf q},Q)$ is given by
\begin{equation}\label{22}
 \mathcal{F}_{\bf{x}}({\bf q},Q) = \mathbb{E}\left\{ e^{i \sum_{k=1}^N q_kf_k+i\sum_{k,l=1}^N\frac{\partial f_k}{\partial x_l}Q_{lk}}\right\}
\end{equation}
\begin{equation}\label{23}
= e^{-\frac{1}{2}\sum_{k,p=1}^N q_kq_p\left\langle f_kf_p\right\rangle-\frac{1}{2}\sum_{k,l,p=1}^N
 \left\langle f_k\frac{\partial f_p}{\partial x_l}\right\rangle q_kQ_{lp}-\frac{1}{2}\sum_{k,l=1}^N \sum_{p,n=1}^NQ_{lk}Q_{np}  \left\langle\frac{\partial f_k}{\partial x_l}\frac{\partial f_p}{\partial x_n}\right\rangle}
\end{equation}
where we have used the gaussian nature of $f_{k}({\bf x})$ and $\frac{\partial}{\partial x_l}f_{k}({\bf x})$ and introduced the short-hand
notation $<AB>$ for the corresponding covariances. The explicit form for the covariances can be easily found from (\ref{6}) by differentiating
and eventually setting ${\bf x}={\bf x'}$:
\begin{equation}\label{24a}
 \left\langle f_k\frac{\partial f_p}{\partial x_l}\right\rangle=\delta_{kl}\frac{x_p}{N}\Phi_2\left(\frac{\bf{x}^2}{N}\right)+
 \delta_{kp}\frac{x_l}{N}\Phi'_1\left(\frac{\bf{x}^2}{N}\right)+
 \frac{x_p x_l x_k}{N^2}\Phi'_2\left(\frac{\bf{x}^2}{N}\right)
 \end{equation}
 and
 \begin{equation}\label{24b}
 \left\langle \frac{\partial f_k}{\partial x_n}\frac{\partial f_p}{\partial x_l}\right\rangle=
  \delta_{pn}\delta_{kl}\frac{1}{N}\Phi_2\left(\frac{\bf{x}^2}{N}\right)+\delta_{kp}\delta_{ln}\frac{1}{N}\Phi'_1\left(\frac{\bf{x}^2}{N}\right)
 +\delta_{kp}\frac{x_lx_n}{N^2}\Phi''_1\left(\frac{\bf{x}^2}{N}\right)
  \end{equation}
  \[+\left[\delta_{kl}\frac{x_px_n}{N^2}+\delta_{pn}\frac{x_lx_k}{N^2}+\delta_{ln}\frac{x_px_k}{N^2}\right]
\Phi'_2\left(\frac{\bf{x}^2}{N}\right)+
\frac{x_px_lx_kx_n}{N^3}\Phi_2''\left(\frac{\bf{x}^2}{N}\right)
\]
which implies
\begin{equation}\label{25a}
\sum_{k,p=1}^N q_kq_p\left\langle f_kf_p\right\rangle={\bf q}^2\Phi_1\left(\frac{\bf{x}^2}{N}\right)+\frac{1}{N}\left({\bf q}\cdot {\bf x}\right)^2\Phi_2\left(\frac{\bf{x}^2}{N}\right)\,,
\end{equation}
\begin{equation}\label{25b}
\sum_{k,l,p=1}^N
 \left\langle f_k\frac{\partial f_p}{\partial x_l}\right\rangle q_kQ_{lp}=\frac{1}{N}\left({\bf q}^TQ{\bf x}\right)\left[\Phi'_1\left(\frac{\bf{x}^2}{N}\right)+\Phi_2\left(\frac{\bf{x}^2}{N}\right)\right]+\frac{1}{N^2}\left({\bf q}\cdot {\bf x}\right)\left({\bf q}^TQ{\bf x}\right)\Phi'_2\left(\frac{\bf{x}^2}{N}\right)
 \end{equation}
 where we find it convenient to use here and henceforth the notation ${\bf q}^TQ{\bf x}$ for the inner product of the vectors ${\bf q}$  and $ Q{\bf x}$. Further
 \begin{equation}\label{25c}
 \sum_{k,p=1}^N \sum_{p,n=1}^NQ_{lk}Q_{np}  \left\langle\frac{\partial f_k}{\partial x_l}\frac{\partial f_p}{\partial x_n}\right\rangle=
 \frac{1}{N}\Phi'_1\left(\frac{\bf{x}^2}{N}\right)\mbox{Tr}\left(QQ^T\right)+\frac{1}{N}\Phi_2\left(\frac{\bf{x}^2}{N}\right)
 \mbox{Tr}\left(Q^2\right)
\end{equation}
\[
+ \frac{1}{N^2}\Phi''_1\left(\frac{{\bf x}^2}{N}\right) \left({\bf x}^TQQ^T{\bf x}\right)+\frac{1}{N^2}\Phi'_2 \left({\bf x}^2\right)\left[\left({\bf x}^TQQ^T\bf{x}\right)+2\left({\bf x}^TQ^2{\bf x}\right)\right]+
\frac{1}{N^3}\Phi''_2\left(\frac{\bf{x}^2}{N}\right) \left({\bf x}^TQ{\bf x}\right)^2
\]
The key observation is that the expressions (\ref{25a})-(\ref{25c}) remain invariant under the set of simultaneous transformations
\begin{equation}\label{26}
{\bf x}\to O{\bf x}, \quad {\bf q}\to O{\bf q}, \quad Q\to OQO^T \quad \mbox{where} \quad OO^T={\bf 1}_N
\end{equation}
 so that the matrix $O$ is any $N\times N$ orthogonal: $O\in O(N)$.
This via (\ref{23}) implies that
\begin{equation}\label{26a}
\mathcal{F}_{\bf{x}}({\bf q},Q)=\mathcal{F}_{O\bf{x}}(O{\bf q},OQO^T), \quad \forall O\in O(N)
\end{equation}
which in view of the invariance of the integration measure $d{\bf q} dQ$ translates via (\ref{21}) into
\begin{equation} \label{26b}
\mathcal{P}_{\bf{x}}({\bf f},K)=\mathcal{P}_{O\bf{x}}(O{\bf f},OKO^T), \quad \forall O\in O(N)
\end{equation}
Finally, substituting (\ref{26b}) to (\ref{20}) and noticing that also
\begin{equation}\label{27}
\left| \det\left( \begin{array}{cc} K -\lambda{\bf 1}_N   & -{\bf x} \\ 2{\bf x}^T & 0 \end{array}
\right) \right|=\left| \det\left( \begin{array}{cc} OKO^T -\lambda{\bf 1}_N   & -O{\bf x} \\ 2\left(O{\bf x}\right)^T & 0 \end{array}
\right) \right|, \quad \forall O\in O(N)
\end{equation}
makes it evident that
\begin{equation} \label{28}
\mathcal{I}\left(\bf{x},\lambda\right)=\mathcal{I}\left(O\bf{x},\lambda\right), \quad \forall O\in O(N).
\end{equation}
 We therefore conclude that $\mathcal{I}(\bf{x},\lambda)=\mathcal{I}(r,\lambda)$, i.e. depends only on the length $r=|{\bf x}|$ of the vector ${\bf x}$ but not on its direction.
Hence, we are free to choose the vector ${\bf x}$ in the form
\begin{equation}\label{29}
{\bf x}=r{\bf e}_1, \quad \mbox{where} \quad {\bf e}_1=\left(\begin{array}{c}1\\0\\.\\.\\.\\0\end{array}\right).
\end{equation}
In what follows it turns out to be convenient to decompose the matrix $K$ as
\begin{equation}\label{30}
K=\left(\begin{array}{cc}k_{11} & {\bf k}_1^T\\{\bf k}_2 & \tilde{K}\end{array}\right).
\end{equation}
where ${\bf k}_1$ and ${\bf k}_2$ are $(N-1)$ component vectors, and $\tilde{K}$ is a real matrix of size $(N-1)\times (N-1)$.
By expanding the determinant it is then easy to check that
\begin{equation}\label{31}
\left| \det\left( \begin{array}{cc} K -\lambda{\bf 1}_N   & -r{\bf e_1} \\ 2r\bf{e_1} & 0 \end{array}
\right) \right|=2r^2\left| \det\left( \begin{array}{c} \tilde{K} -\lambda{\bf 1}_{N-1}   \end{array}
\right) \right|\,.
\end{equation}
 Further, the invariance of $\mathcal{I}\left(\bf{x},\lambda\right)$ under $O(N)$ rotations allows to perform the integration in (\ref{20a}) as
 \begin{equation}\label{32a}
\mathbb{E}\left\{\mathcal{N}_{tot}\right\} \!=\!\frac{2(\pi N)^{N/2}}{\Gamma\left(\frac{N}{2}\right)}\int_{[\alpha,\beta]} \, \mathcal{I}(N,\lambda)\, d\lambda
\end{equation}
where
\begin{equation}\label{32b}
\mathcal{I}(N,\lambda) = \int\,\int\,\! \tilde{\mathcal{P}}\left({\bf f},\tilde{K}\right) e^{-\frac{1}{2\sigma^2}\left(\lambda\,\sqrt{N}{\bf e}_1-{\bf f}\right)^2} \left| \det\left( \begin{array}{c} \tilde{K} -\lambda{\bf 1}_{N-1}   \end{array}
\right) \right|\, d\tilde{K} \,  d{\bf f}\, ,
\end{equation}
and we defined
\begin{equation}\label{32c}
\tilde{\mathcal{P}}\left({\bf f},\tilde{K}\right)=\int \mathcal{P}_{\sqrt{N}{\bf e}_1}\left({\bf f},\tilde{K}\right)\,dk_{11}\,d{\bf k}_1\,d{\bf k}_2
\end{equation}
Employing now the Fourier-transform (\ref{21}) and the decomposition (\ref{30}) for $K$ it is easy to see that
\begin{equation}\label{32d}
\tilde{\mathcal{P}}\left({\bf f},\tilde{K}\right)=\int \tilde{\mathcal{F}}\left({\bf q},\tilde{Q}\right) e^{-i{\bf q}^T{\bf f}-i\mbox{\small Tr}\left(\tilde{K}\tilde{Q}\right)} \frac{d{\bf q}}{(2\pi)^N}\frac{d \tilde{Q} }{(2\pi)^{(N-1)^2}}
\end{equation}
where we have used the decomposition
 \begin{equation}\label{33}
Q=\left(\begin{array}{cc}q_{11} & {\bf p}_1^T\\{\bf p}_2 & \tilde{Q}\end{array}\right),
\end{equation}
with ${\bf p}_1,{\bf p}_2$ being $(N-1)-$component vectors, and the function $\tilde{\mathcal{F}}\left({\bf q},\tilde{Q}\right)$
is obtained from $\mathcal{F}_{\bf{x}}({\bf q},Q)$ by replacing  ${\bf x}\to \sqrt{N}{\bf e}_1$ and setting all the the variables $Q_{11},{\bf p}_1,{\bf p_2}$ to zero. Performing the required substitutions in  (\ref{25a})-(\ref{25c}) and   further decomposing
 the vector ${\bf q}$, in (\ref{21}) as  ${\bf q}^T=(q_1,\tilde{\bf q})$ with $\tilde{\bf q}$ being $(N-1)-$component vector yields after straightforward manipulations a simple expression
  \begin{equation}\label{34}
 \tilde{\mathcal{F}}\left({\bf q},\tilde{Q}\right)=\exp{\left\{-\frac{1}{2N}\left[\Phi_1(1)\left(q_1^2+\tilde{\bf q}^2\right)+
 \Phi_2(1)\,q_1^2+\Phi'_1(1)\mbox{Tr}\left(\tilde{Q}\tilde{Q}^T\right)+\Phi_2(1)\mbox{Tr}\left(\tilde{Q}^2\right)\right] \right\}}
  \end{equation}
With such expressions in hand, the integrals in (\ref{32d}) are straightforward to perform. In particular, we have
   \begin{equation}\label{35a}
\mathcal{P}(\tilde{K})=\int e^{-i\mbox{\small Tr}\left(\tilde{K}\tilde{Q}\right)-\frac{1}{2N}\left[\Phi'_1(1)\mbox{\small Tr}\left(\tilde{Q}\tilde{Q}^T\right)+\Phi_2(1)\mbox{\small Tr}\left(\tilde{Q}^2\right)\right] }\, \frac{d \tilde{Q} }{(2\pi)^{(N-1)^2}}
  \end{equation}
\begin{equation}\label{35b}
=C_K\,\exp\left\{-\frac{N}{2(\Phi'_1(1)^2-\Phi_2(1)^2)}\left[\Phi'_1(1)\mbox{\small Tr}\left(\tilde{K}\tilde{K}^T\right)-\Phi_2(1)\mbox{\small Tr}\left(\tilde{K}^2\right)\right]\right\}
 \end{equation}
where
\begin{equation}\label{35b}
C_K=\frac{1}{\left[\frac{2\pi}{N}\left(\Phi'_1(1)+\Phi_2(1)\right)\right]^{\frac{(N-1)^2}{2}}\left[\Phi'_1(1)-\Phi_2(1)\right]^{\frac{(N-1)(N-2)}{2}}}
 \end{equation}
Further integrating over  $ \tilde{\bf q}$ and $q_1$, and performing the Gaussian integral over ${\bf f}$ in  (\ref{32b}) we arrive at
\begin{equation}\label{36a}
\mathcal{I}(N,\lambda) = C_q e^{-\frac{N \lambda^2}{2\left(\sigma^2+\Phi_1(1)+\Phi_2(1)\right)}}\int\,\! \mathcal{P}\left(\tilde{K}\right)
 \left| \det\left( \begin{array}{c} \tilde{K} -\lambda{\bf 1}_{N-1}   \end{array}
\right) \right|\, d\tilde{K} \, \, ,
\end{equation}
with
\begin{equation}\label{36b}
C_q=\frac{1}{\sqrt{2\pi\left(\sigma^2+\Phi_1(1)+\Phi_2(1)\right)}\left[2\pi\left(\sigma^2+\Phi_1(1)\right)\right]^{\frac{(N-1)}{2}}}
 \end{equation}
Finally, substituting (\ref{36a}-\ref{36b}) into (\ref{32a}),  introducing the $(N-1)\times (N-1)$ matrices $X$ via the rescaling $\tilde{K}=\sqrt{\frac{\Phi'_1(1)}{N}} \, X$ and simultaneously rescaling $\lambda\to \sqrt{\Phi'_1(1)}\lambda$ one arrives at
(\ref{A}).
\end{proof}

Now we can proof {\bf Theorem (\ref{Thm1})}.
\begin{proof}
Assume for simplicity that $(\alpha,\beta)=(-\infty,\infty)$. The representation (\ref{A}) for the mean total number of equilibria is immediately converted to (\ref{12}) by employing the following relation between the  density of real eigenvalues $\rho_{N}^{(r)}(x)$ of the $N\times N$ matrices $X$
from the elliptic ensemble (\ref{11}) and the expected value of the modulus of the characteristic determinant featuring in  (\ref{A}):
\begin{equation}\label{37}
\mathbb{E}\left\{\left| \det\left( \begin{array}{c} X -\lambda\sqrt{N}{\bf 1}_{N-1}   \end{array}
\right) \right|\right\}_{X_{N-1}}=2\sqrt{1+\tau}\frac{(N-1)!}{(N-2)!!} e^{N\frac{\lambda^2}{2(1+\tau)}}\rho_{N}^{(r)}(\lambda\sqrt{N})\,
\end{equation}
For the limiting case $\tau=0$ this relation appeared originally in \cite{EKS}, and for $\tau=1$, when all eigenvalues of $X$ are real, by a different method in  \cite{Fyo04}. The general proof for any $\tau\in (-1,1)$ can be found in \cite{FK2016}.
\end{proof}

\section{Asymptotic analysis for $N\gg 1$: proof of Propositions (2.5)-(2.8)}

The density function $\rho_{N}^{(r)}(x)$ is known in the closed form in terms of Hermite polynomials \cite{ForNag2008}.
Assuming for simplicity that $N$ is even, one has $\rho_{N}^{(r)}(x)=\rho_{N}^{(r),1}(x)+\rho_{N}^{(r),2}(x)$ where
\begin{equation}\label{38a} \rho_{N}^{(r),1}(x)=
\frac{1}{\sqrt{2\pi}} \, \sum_{k=0}^{N-2}\, \frac{\big|\psi^{(\tau)}_k(x)\big|^2}{k!},
\end{equation}
and
\begin{equation}\label{38b}
 \rho_{N}^{(r),2}(x)=
\frac{1}{\sqrt{2\pi}(1+\tau)(N-2)!}\,  \psi^{(\tau)}_{N}(x) \int_0^{x}\psi^{(\tau)}_{N-2}(u)\,du.
\end{equation}
Here $\psi^{(\tau)}_k(x)=e^{-\frac{x^2}{2(1+\tau)}}h^{(\tau)}_k(x)$ and $h^{(\tau)}_k(x)$ are rescaled Hermite polynomials,
$h^{(\tau)}_{k}(x)=\frac{1}{\sqrt{\pi}}\int_{-\infty}^{\infty}e^{-t^2}
\left(x+it\sqrt{2\tau}\right)^k\,dt$.

Based on this expression, all relevant asymptotics of $\rho_{N}^{(r)}(\lambda\sqrt{N})$ for large $N\gg 1$ were worked out in \cite{ForNag2008} and \cite{FK2016}.  This allows one to carry out an asymptotic evaluation of the integral in (\ref{12}) in various regimes of parameters $\tau$ and $b^2$ in the limit $N\to \infty$. As the analysis goes very much in parallel to one presented in \cite{FK2016} we only sketch below its main steps working out explicitly the differences and necessary modifications.

 Fixing the parameter $-1<\tau<1$ and considering the spectral parameter $|x|$ in the bulk of the spectrum for elliptic ensemble, i.e., for $|x|<(1+\tau)\sqrt{N}$, one finds that the contribution of (\ref{38a}) to $\rho_{N}^{(r)}(x)$  is dominant and, to the leading order in $N$,
\begin{equation}\label{39a}
\left. \rho_{N}^{(r)}(\lambda\sqrt{N})\right|_{|\lambda|<1+\tau}
=\frac{1}{\sqrt{2\pi(1-\tau^2)}}.
\end{equation}
At the same time, outside the bulk for $|x|>(1+\tau)\sqrt{N}$ both (\ref{38a}) and (\ref{38b}) yield exponentially small contributions to $\rho_{N}^{(r)}(x)$, with  (\ref{38b}) being dominant. Our evaluation yields in this case the leading term as
  \begin{equation}\label{39b}
\left.  \rho_{N}^{(r)}(\lambda\sqrt{N}) \right|_{\lambda>(1+\tau)}= Q(\lambda) \exp{-N\Psi(\lambda)},
\end{equation}
where
 \begin{equation}\label{39b1}
Q(\lambda)=\left[\frac{N}{2\pi(1+\tau)}\frac{1}{\sqrt{\lambda^2-4\tau}(\lambda+\sqrt{\lambda^2-4\tau})}\right]^{1/2}\, ,
\end{equation}
 \begin{equation}\label{39b2}
\Psi(\lambda)=-\frac{1}{2}+\frac{\lambda^2}{2(1+\tau)}-\frac{1}{8\tau}(\lambda-\sqrt{\lambda^2-4\tau})^2-
\ln{\frac{\lambda+\sqrt{\lambda^2-4\tau}}{2}}\, .
\end{equation}
 Note that the corresponding expressions [24]-[25] presented in \cite{FK2016} contained several misprints, in particular the constant term in (\ref{39b2}) was missing and the spurous factor $\sqrt{\tau}$ appeared under the last logarithm), though the correct expressions (\ref{39b1}-\ref{39b2}) were used for actual calculations.

  The information above is sufficient to  verify the statements of {\bf Proposition (\ref{Prop1})}.

\begin{proof}

It is evident, that as long as the parameter $B$ defined in (\ref{12a}) is positive, that is as long as $b^2<1$, the leading asymptotics of the integral in (\ref{12}) is obtained by substituting the density $\rho_{N}^{(r)}(\lambda\sqrt{N})$ in the form  (\ref{39a}). Elementary asymptotic evaluation of the integral by the Laplace method yields (\ref{13}) thus proving the first part of the {\bf Proposition (\ref{Prop1})}.

In the opposite case of $B<0$ , that is $b^2>1$, the asymptotics is determined by a competition between growing exponential factor in the integrand and
the exponential decrease in the density  $\rho_{N}^{(r)}(\lambda\sqrt{N})$  in the region $|\lambda|>1+\tau$ described by the formulae (\ref{39b}-\ref{39b2}).  It is then evident (and is confirmed by direct calculation below) that the point $\lambda_{*}$ of maximum dominating the integrand
belongs to the domain $|\lambda|>1+\tau$, so for our purposes we need to evaluate the asymptotics of the integral
 \begin{equation}\label{40}
\mathcal{I}=\int_{1+\tau}^{\infty} Q(\lambda)e^{N \mathcal{L}(\lambda)}, \quad  \mathcal{L}(\lambda)=-\frac{B}{4}\lambda^2-\Psi(\lambda)
\end{equation}
Applying the Laplace method we seek for the maximum of $\mathcal{L}(\lambda)$ and to that end consider the stationary points satisfying
\begin{equation}\label{40a}
\frac{d}{d\lambda } \mathcal{L}(\lambda)\left|_{\lambda_*}\right.=0, \quad \mbox{where} \quad \frac{d}{d\lambda } \mathcal{L}(\lambda)=-\lambda\,\frac{1}{b^2+\tau}+\frac{1}{2\tau}(\lambda-\sqrt{\lambda^2-4\tau})
\end{equation}
which yields
\begin{equation}\label{40b}
\lambda_*=b+\frac{\tau}{b}
\end{equation}
which indeed satisfies $\lambda_*>1+\tau$ due to $b^2>1\ge \tau$. Further we find from (\ref{40b}) that
${\frac{\lambda_*+\sqrt{\lambda_*^2-4\tau}}{2}}=b$ which in turn implies after straightforward calculation:
\begin{equation}\label{41}
 \mathcal{L}(\lambda)\left|_{\lambda_*}\right.=\ln{b}, \quad \frac{d^2}{d\lambda^2 } \mathcal{L}(\lambda)\left|_{\lambda_*}\right.=-\frac{2b^2}{(b^2+\tau)(b^2-\tau)}\,.
 \end{equation}
Now applying the standard Laplace method to (\ref{40}) one finds that
\begin{equation}\label{42}
\mathcal{I}= Q(\lambda_*)e^{N \mathcal{L}(\lambda_*)}\frac{1}{\sqrt{N\frac{d^2}{d\lambda^2 } \mathcal{L}(\lambda)\left|_{\lambda_*}\right.}}=b^{N-1}\sqrt{\frac{b^2+\tau}{4N(1+\tau)}}
\end{equation}
which when combined with the same contribution to (\ref{12}) from the integration domain $\lambda<-(1+\tau))$ and multiplied
with the correct pre-factor from (\ref{12}) finally yields (\ref{13a}), thus completing the proof of the {\bf Proposition (\ref{Prop1})}.
Note that the last part of the calculation extends without change also to the boundary case $\tau=1$.

\end{proof}

Let us now sketch the proof of the {\bf Proposition (\ref{Prop2})}

\begin{proof}:  Fix any $|\tau|<1$ and for $N\to \infty$ scale the parameter $b$ with $N$ as $b^2=1-\frac{\gamma}{N}$, with the parameter $-\infty<\gamma<\infty$ being fixed. In this regime the product $NB\to \frac{\gamma}{(1+\tau)^2}$ so remains of the order of unity, hence the leading contribution to the integral (\ref{12}) comes from
 the integration domain $|\lambda|<1+\tau$ where the density of real eigenvalues is of the order of unity and is given asymptotically by (\ref{39a}).
 After dividing (\ref{12}) by $\sqrt{N}$ one can see that the remaining factors have a well-defined large-$N$ limit. In particular  $b^{-N}\to e^{\gamma/2}$,
   and after rescaling $\lambda\to (1+\tau)\lambda$ one straightforwardly arrives at (\ref{14}).
\end{proof}

Similarly, we sketch the proof of the {\bf Proposition (\ref{Prop3})}.

\begin{proof}
Fix any $|\tau|<1$ and for $N\to \infty$ scale the parameter $b$ with $N$ as $b^2=1+\frac{\kappa}{\sqrt{N}}$, with the parameter $\kappa>0$ being fixed.
It is clear that the main contribution to the integral (\ref{12}) comes in this case from the vicinity of the spectral edges $\lambda=\pm(1+\tau)$,
and it is enough to consider the right edge. In such a vicinity it turns out to be convenient to parametrize the spectral parameter as
\begin{equation}\label{43}
\lambda=1+\tau+\frac{\zeta}{\sqrt{N}}\sqrt{1-\tau^2}
\end{equation}
where the variable $-\infty<\zeta<\infty$ is considered to be of the order of unity. The behaviour of the density
$\rho_{N}^{(r)}(\lambda\sqrt{N})$ in that scaling regime was determined in the paper by Forrester and Nagao \cite{ForNag2008}
and is given precisely by the 'edge density' (\ref{15a}). In such a regime one also finds by direct computation the following large$-N$ limiting behaviour:
\begin{equation}\label{43a}
\lim_{N\to \infty} b^{-N}e^{-\frac{\lambda^2}{4}BN}=\exp\left\{-\frac{\kappa^2}{4}\frac{1-\tau}{1+\tau}+\kappa\zeta\sqrt{\frac{1-\tau}{1+\tau}}\right\}
\end{equation}
Using these facts it is immediate to arrive to the formula  (\ref{15}).

\end{proof}

Finally, the proof of the {\bf Proposition (\ref{Prop4})} goes along the following lines.

\begin{proof}
 Let $\tau=1-\frac{u^2}{N}, \, 0\le u<\infty$ and consider $b^2<1$. The real eigenvalues of the matrix $X$ for $N\to \infty$ with probability tending to one belong to the interval $\lambda\in (-2,2)$. Correspondingly, the mean density $\rho_{N}^{(r)}(\lambda\sqrt{N})$ in this scaling limit
 was found in \cite{FK2016} and is equal to the leading order to
 \begin{equation} \label{44}
 \rho_{N}^{(r)}(\lambda\sqrt{N})=\sqrt{N}{\pi}\int_0^{\sqrt{1-\frac{\lambda^2}{4}}} e^{-u^2t^2}\,dt, \quad |\lambda|<2
 \end{equation}
 and is exponentially small outside that interval. Moreover, in this regime we further have to the leading order $B=\frac{1-b^2}{2(1+b^2)}>0$ and the integral (\ref{12}) is dominated by small values $\lambda\sim N^{-1/2}$. Extracting in a standard way the leading contribution to the integral
 reproduces (\ref{17})

\end{proof}

\end{document}